\documentclass[11pt, twoside, letter]{article}
\usepackage[margin = 1in]{geometry}
\usepackage[utf8]{inputenc}
\usepackage{amsmath, amssymb, amsfonts}
\usepackage{ifthen}
\usepackage[boxed, linesnumbered]{algorithm2e}
\usepackage{comment}

\usepackage{pgf}
\usepackage{tikz}
\usetikzlibrary{arrows,automata}

\usepackage{microtype}

\usepackage{amsthm}
\usepackage[unicode,psdextra,colorlinks=true,citecolor=blue,bookmarksnumbered=true,hyperfootnotes=false,linkcolor=green!40!black]{hyperref}
\usepackage{multirow, multicol}

\usepackage{subfigure, graphicx} 
 \usepackage{bmpsize}
\usetikzlibrary{arrows,shapes}

\tikzset{circle/.pic={
\node[circle, aspect=1, draw, minimum size=0.3cm, text width=0.2cm] () at (0,0) {\tikzpictext};
}} 
\usepackage{microtype}
\usepackage{pgfplots}
\pgfplotsset{compat=1.10}

\newtheorem{theorem}{Theorem}[section]
\newtheorem{corollary}[theorem]{Corollary}
\newtheorem{lemma}[theorem]{Lemma}

\newtheorem{definition}[theorem]{Definition}
\newtheorem{remark}[theorem]{Remark}

\newcommand{\defeq}{\stackrel{\textup{def}}{=}}

\newcommand{\ie}{{\it i.e., }}

\renewcommand{\vec}[1]{\mathbf{#1}}

\newcommand{\etal}{{\em et al.~}}

\begin{document}
\title{Multiplicative Weights Update with Constant Step-Size in Congestion Games: Convergence, Limit Cycles and Chaos}


\author{Gerasimos Palaiopanos\\SUTD\\gerasimosath@yahoo.com
\and Ioannis Panageas\\MIT \& SUTD\\ioannisp@mit.edu
\and Georgios Piliouras\\SUTD\\georgios@sutd.edu.sg}


\date{}
\maketitle
\begin{abstract}
The Multiplicative Weights Update (MWU) method is a ubiquitous meta-algorithm that works as follows:
 A distribution is maintained on a certain 
  set, and at each step the probability assigned to element 
   $\gamma$ is multiplied
 by $(1 -\epsilon C(\gamma))>0$ where $C(\gamma)$ is the ``cost" of element 
 $\gamma$ and then rescaled to ensure that the new values form a distribution.
We analyze MWU in congestion games where agents use \textit{arbitrary admissible
constants} as 
 learning rates $\epsilon$
 and prove convergence to \textit{exact Nash equilibria}.
Our proof leverages a novel connection between MWU and the Baum-Welch algorithm, the standard
instantiation of the Expectation-Maximization (EM) algorithm for hidden Markov models (HMM).
 Interestingly, this convergence result does not carry over to the nearly homologous MWU variant where at each step the probability assigned to element $\gamma$ 
 is multiplied by $(1 -\epsilon)^{C(\gamma)}$ even for the most innocuous case of two-agent, two-strategy load balancing games, where such dynamics can provably lead to limit cycles
 or even chaotic behavior.
\end{abstract}




\section{Introduction}

The Multiplicative Weights Update (MWU) is a ubiquitous meta-algorithm with numerous applications in different fields~\cite{Arora05themultiplicative}.
It is particularly useful in game theory due to its regret-minimizing properties \cite{Fudenberg98,Cesa06}.
It is typically introduced in two nearly identical variants, the one in which at each step the probability assigned to action $\gamma$
is multiplied by $(1-\epsilon C(\gamma))$  and the one in which it is multiplied by $(1-\epsilon)^{C(\gamma)}$ where $C(\gamma)$ is the cost of action $\gamma$. We will refer to the first as the linear variant, $\text{MWU}_\ell$, and the second  as the exponential, $\text{MWU}_e$. In the literature there is little distinction between these two variants as both carry the same advantageous regret-minimizing property. 
It is also well known that in order to achieve sublinear regret, the learning rate $\epsilon$ must be decreasing as time progresses.
This  constraint raises a natural question: Are there interesting classes of games where MWU behaves well without the need to  fine-tune its learning rate?

A natural setting to test the learning behavior of MWU with constant learning rates $\epsilon$ is the class of congestion games.
Unfortunately, even for the most innocuous instances of congestion  games $\text{MWU}_e$ fails to converge to equilibria.
For example, even in the simplest case of two balls two bins games,\footnote{$n$ balls $n$ bin games are symmetric load balancing games with $n$ agent and $n$ edges/elements each with a cost function of c(x)=x. We normalize costs equal to $c(x)=x/n$ so that they lie in $[0,1]$.} $\text{MWU}_e$ with $\epsilon=1-e^{-10}$ is shown to converge to a limit cycle of period $2$ for
infinitely many initial conditions (
Theorem \ref{thm:cycle}). If the cost functions of the two edges are not identical then we create instances of two player load balancing games such that $\text{MWU}_e $ has periodic orbits of length $k$ for all $k>0$, as well as uncountable many initial conditions which never settle on any periodic orbit but instead exhibit an irregular behavior known as Li-Yorke chaos
(
Corollary \ref{coro:chaos}).

The source of these problems is exactly the large, fixed learning rate $\epsilon$, \textit{e.g.}, $\epsilon\approx 1$ for costs in $[0,1]$.
Intuitively, the key aspect of the problem can be captured by (simultaneous) best response dynamics. If both agents start from the same edge and best-respond simultaneously they will land on the second edge which now has a load of two. In the next step they will both jump  back to the first edge and this motion will be continued perpetually.  Naturally, $\text{MWU}_e$ dynamics are considerably more intricate as they evolve over mixed strategies and allow for more complicated non-equilibrium behavior but the key insight is correct.
Each agent has the right goal, decrease his own cost and hence the potential of the game, however, as they pursue this goal too aggressively they cancel each other's gains and lead to unpredictable non-converging behavior.

In a sense, the cautionary tales above agree with our intuition. Large, constant learning rates $\epsilon$ nullify the known performance guarantees of MWU. We \textit{should} expect erratic behavior in such cases. The typical way to circumvent these problems is through careful monitoring and possibly successive halving of the $\epsilon$ parameter, a standard technique in the MWU literature. In this paper, we explore an alternative, cleaner, and surprisingly elegant solution to this problem. \textit{We show that applying $\text{MWU}_\ell$, the linear variant of MWU, suffices to guarantee convergence in all congestion games.}



\subsection*{Our contribution.}

Our key result is the proof of convergence of $\text{MWU}_\ell$ in congestion games. The main technical contribution is a proof that
the potential of the mixed state is always strictly decreasing along any nontrivial trajectory (Theorem \ref{mainthm}).
This result holds for all congestion games, irrespective of the number of agents or the size, topology of the strategy sets. Moreover, each agent $i$ may be applying different learning rates $\epsilon_i$. The only restriction on the set of allowable learning rates $\epsilon_i$ is that for each agent the multiplicative factor $(1-\epsilon_i C_i(\vec{s}))$ should be positive for all strategy outcomes $\vec{s}$.\footnote{This is an absolutely minimal restriction so that the denominator of $\text{MWU}_\ell$ cannot become equal to zero.}  Arguing convergence to equilibria for all initial conditions (Theorem \ref{Qincreasing}) and further, convergence to Nash equilibria for all interior initial conditions (Theorem \ref{thm:convergencenash}) follows. Proving that the potential always decreases (Theorem \ref{mainthm}) hinges upon discovering a novel interpretation of MWU dynamics. Specifically, we show that the class of dynamical  systems derived by applying $\text{MWU}_\ell$ in congestion games
is a special case of a convergent class of dynamical systems introduced by Baum and Eagon (Theorem \ref{thm:B-E} \cite{Baum_Eagon}). The most well known member of this class is the classic Baum-Welch algorithm, the standard instantiation of the Expectation-Maximization (EM) algorithm for hidden Markov models (HMM). Effectively, the proof of convergence of both these systems boils down to a proof of membership to the same class of Baum-Eagon systems (see section \ref{sec:BEW} for more details on these connections).

We conclude by providing simple congestion games where $\text{MWU}_e$  fails to converge. The main technical contribution of this section is proving convergence to a limit cycle, specifically a periodic orbit of length two,  for the simplest case of two balls two bins games for
infinitely many initial conditions (Theorem \ref{thm:cycle}). After normalizing costs to lie in $[0,1]$, i.e. $c(x)=x/2$, we prove that almost all symmetric non-equilibrium initial conditions converge to a unique limit cycle when both agents use learning rate  $\epsilon=1-e^{-10}$.
In contrast, since $1-\epsilon \cdot C(\vec{s})\geq1-(1-e^{-10})1=e^{-10}>0$, $\text{MWU}_\ell$ successfully converges to equilibrium. Establishing chaotic behavior for the case of edges with different cost functions is rather straightforward in comparison (Corollary \ref{coro:chaos}).
The key step is to exploit symmetries in the system to reduce it to a single dimensional one and then establish the existence of a periodic orbit of length three. The existence of periodic orbits of any length as well as chaotic orbits then follows from the Li-Yorke theorem \ref{thm:liyorke} \cite{liyorke} (see section \ref{sec:chaos} for background on chaos and dynamical systems).







\section{Related Work}

\textbf{Congestion/potential games:}
Congestion games are amongst the most well known and thoroughly studied class of games. 
Proposed in \cite{rosenthal73} and isomorphic to potential games \cite{potgames}, they have been successfully employed in myriad modeling problems.
The study the price of anarchy, i.e. efficiency guarantees for equilibria, in congestion games is arguably amongst the most developed areas within algorithmic game theory, e.g., 
\cite{KoutsoupiasP99WorstCE,roughgarden2002bad,christodoulou,Fotakis2005226,Schafer10,Roughgarden09}.

 It is common knowledge 
 that better-response dynamics in congestion games converge. In these dynamics, in every round, exactly one agent deviates to a better strategy. If two or more agents move at the same time then convergence is not guaranteed.  
Despite the numerous positive convergence results for concurrent  dynamics in congestion games, e.g., \cite{Fotakis08,Berenbrink:2007:DSL:1350525.1350533,Ackermann:2009:CID:1582716.1582732,hoo,kleinberg2011load,cara,sincla}, we know of no prior work establishing such a clean, \textit{deterministic} convergence result to \textit{exact} Nash equilibria for general atomic congestion games. MWU has also been studied  in congestion games. In \cite{Kleinberg09multiplicativeupdates}  randomized variants of the exponential version of the MWU are shown to converge w.h.p. to pure Nash equilibria as long as the learning rate $\epsilon$ is small enough. In contrast our positive results for linear $MWU_{\ell}$ hold deterministically and for all learning rates. Our paper establishes that these results cannot be extended to the exponential $MWU_{e}$ even for two balls two bin games.

\textbf{Multiplicative Weights Update and connections:}
The multiplicative weights update method is a  widely used meta-algorithm. From the perspective of online learning it belongs to the class of regret minimizing algorithms. As a result it is widely applicable in algorithmic game theory, as the time average behavior of MWU leads to (approximate) coarse correlated equilibria (CCE) for which price of anarchy guarantees apply~\cite{Roughgarden09}. 
In the last couple of years several theoretical results have been proved on the intersection of computer science, learning and evolution  for which MWU was the linking component.  In \cite{ITCS:DBLP:dblp_conf/innovations/ChastainLPV13,PNAS2:Chastain16062014} Chastain et al. show that standard models of haploid evolution can be directly interpreted as MWU dynamics  \cite{Hofbauer98} employed in coordination games. Meir and Parkes \cite{Reshef15}, Mehta et al. \cite{ITCS15MPP} have shed more light on these connections.


\textbf{Non-convergent dynamics:}  Outside the class of congestion games, there exist several negative results in the literature concerning the non-convergence of  MWU and variants thereof. In particular, in \cite{daskalakis10} it was shown that the multiplicative updates algorithm fails to find the unique Nash equilibrium of the $3\times 3$ Shapley game. Similar non-convergent results have been proven for perturbed zero-sum games \cite{Balcan12},
as well as  for the continuous time version of MWU, the replicator dynamics \cite{paperics11,soda14,CRS16}.
The possibility of applying Li-Yorke type arguments for MWU in congestion games with two agents was inspired by a remark in \cite{avramopoulos} for the case of continuum of agents. Our paper is the first to our knowledge where non-convergent MWU behavior in congestion games is formally proven capturing both limit cycles and chaos and we do so in the minimal case of two balls two bin games.

\section{Preliminaries}
\label{Concepts}

\noindent{\bf Notation.}  We use boldface letters, e.g., $\vec{x}$, to denote column vectors (points). For a function $f:\mathbb{R}^m \to \mathbb{R}^m,$ by $f^n$ we denote the composition of $f$ with itself $n$ times, namely $\underbrace{f \circ f \circ \cdots\circ f}_{n \textrm{ times}}$.


\subsection{Congestion Games}
A \emph{congestion game} \cite{rosenthal73} is defined by the tuple $(\mathcal{N}; E;$ $(S_i)_{i \in \mathcal{N}};(c_e)_{e \in E})$ where $\mathcal{N}$ is the set of \emph{agents}, $N = |\mathcal{N}|$, $E$ is a set of \emph{resources} (also known as \emph{edges} or \emph{bins} or \emph{facilities}) and each player $i$ has a set $S_i$ of subsets of $E$ ($S_i \subseteq 2^E$) and $|S_i|\geq 1$. Each strategy $s_i \in S_i$ is a set of edges and $c_e$ is a positive cost (latency) function associated with facility $e$. We use small greek characters like $\gamma, \delta$ to denote different strategies/paths. For a strategy profile $\vec{s} = (s_1,s_2,\dots,s_N)$, the cost of player $i$ is given by $c_i(\vec{s}) = \sum_{e \in s_i} c_e(\ell_e(\vec{s}))$, where $\ell_e(\vec{s})$ is the number of players using $e$ in $\vec{s}$ (the load of edge $e$). 
The potential function is defined to be $\Phi(\vec{s}) = \sum_{e \in E}\sum_{j=1}^{\ell_e(\vec{s})} c_e(j)$.

For each $i \in \mathcal{N}$ and $\gamma \in S_i$,  $p_{i\gamma}$ denotes the probability player $i$ chooses strategy $\gamma$. We denote by $\Delta(S_{i}) = \{\vec{p}\geq \vec{0}: \sum_{\gamma}p_{i\gamma}=1\}$ the set of mixed (randomized) strategies of player $i$ and $\Delta = \times_i \Delta(S_i)$ the set of mixed strategies of all players. We use $c_{i\gamma}=\mathbb{E}_{\vec{s}_{-i}\sim \vec{p}_{-i}} c_i(\gamma,\vec{s}_{-i})$ to denote the expected cost of player $i$ given that he chooses strategy $\gamma$ and $\hat{c}_{i} = \sum_{\delta \in S_i}p_{i\delta}c_{i\delta}$ to denote his expected cost.

\subsection{Dynamical Systems and Chaos}
\label{sec:chaos}

Let $\vec{x}^{(t+1)} = f(\vec{x}^{(t)})$ be a \emph{discrete time} dynamical system with update rule $f:\mathbb{R}^m \to \mathbb{R}^m$.
The point $\vec{z}$ is called a \textit{fixed point} of $f$ if $f(\vec{z}) = \vec{z}$.
A sequence $(f^t(\vec{x}^{(0)}))_{t \in \mathbb{N}}$ is called a \textit{trajectory} or \textit{orbit} of the dynamics with $x^{(0)}$ as starting point.
  A common technique to show that a dynamical system  converges to a fixed point is to construct a function $P : \mathbb{R}^m \to \mathbb{R}$ such that $P(f(\vec{x})) > P(\vec{x})$ unless $\vec{x}$ is a fixed point.
  We call $P$ a \textit{Lyapunov} or \textit{potential} function.

\begin{definition}
$C = \{\vec{z}_1, \ldots, \vec{z}_k\}$ is called a \textit{periodic orbit} of length $k$ if $\vec{z}_{i+1} = f(\vec{z}_i)$ for $1 \leq i \leq k-1$ and $f(\vec{z}_k) = \vec{z}_1$. Each point $\vec{z}_1,\ldots,\vec{z}_k$ is called periodic point of period $k$.  If the dynamics converges to some periodic orbit, we also use the term \textit{limit cycle}.
\end{definition}
Some dynamical systems converge and their behavior can be fully understood and some others have strange, \textit{chaotic} behavior. There are many different definitions for what chaotic behavior and chaos means. In this paper we follow the definition of chaos by Li and Yorke. 
Let us first give the definition of a scrambled set.  Given a dynamical system with update rule $f$, a pair $x$ and $y$ is called ``scrambled" if $\lim_{n\to \infty} \inf |f^n(x) - f^n(y)| =0$ (the trajectories get arbitrarily close) and also  $\lim_{n\to \infty} \sup |f^n(x) - f^n(y)| >0$ (the trajectories move apart). A set $S$ is called ``scrambled" if $\forall x,y\in S$, the pair is ``scrambled".
\begin{definition}[Li and Yorke] A discrete time dynamical system with update rule $f$, $f: X \to X$ continuous on a compact set $X \subset \mathbb{R}$ is called chaotic if (a) for each $k \in \mathbb{Z}^{+}$, there exists a periodic point $p \in X$ of period $k$ and (b) there is an uncountably infinite set $S\subseteq{X}$ that is ``scrambled".
\end{definition}
Li and Yorke proved the following theorem \cite{liyorke} (there is another theorem of similar flavor due to Sharkovskii \cite{sharkovskii}):
\begin{theorem}[Period three implies chaos]\label{thm:liyorke} Let $J$ be an interval and let $F : J \to J$ be continuous. Assume there is a point $a \in J$ for which the points $b = F(a), c = F^2(a)$ and $d = F^3(a)$, satisfy $$d \leq a < b<c \textrm{ (or }d\geq a>b>c\textrm{).}$$
Then
\begin{enumerate} 
\item For every $k=1,2,\dots$ there is a periodic point in $J$ having period $k$.
\item There is an uncountable set $S \subset J$ (containing no periodic points), which satisfies the following conditions:
\begin{itemize}
\item For every $p,q \in S$ with $p \neq q$, $$\lim_{n \to \infty} \sup |F^n(p) - F^n(q)|>0 \textrm{ and } \lim_{n \to \infty} \inf |F^n(p) - F^n(q)|=0.$$
\item For every point $p \in S$ and periodic point $q \in J$, $$\lim_{n\to \infty} \sup |F^n(p) - F^n(q)|>0.$$
\end{itemize}
\end{enumerate}
Notice that if there is a periodic point with period $3$, then the hypothesis of the theorem will be satisfied.
\end{theorem}
\subsection{Baum-Eagon Inequality, Baum-Welch and EM}
\label{sec:BEW}
We start this subsection by stating the Baum-Eagon inequality. This inequality will be used to show that $\textrm{MWU}_{\ell}$ converges to fixed points and more specifically Nash equilibria for congestion games.
\begin{theorem}[Baum-Eagon inequality \cite{Baum_Eagon}]
\label{thm:B-E}
  Let $P(\vec{x})=P\left(\left\{x_{ij}\right\}\right)$ be a polynomial with nonnegative coefficients homogeneous of degree $d$ in its variables $\left\{x_{ij}\right\}$. Let $\vec{x}=\left\{x_{ij}\right\}$ be any point of the domain $D: x_{ij} \geq 0, \sum_{j=1}^{q_i}x_{ij}=1, i=1,2, ..., p, j=1,2, ..., q_i$. For $\vec{x}=\left\{x_{ij}\right\} \in D$ let $\Im(\vec{x})=\Im\left\{x_{ij}\right\}$ denote the point of $D$ whose i, j coordinate is \\
  \begin{equation*}
      \Im(\vec{x})_{ij} = \left( \left. x_{ij} \frac{\partial P}{\partial x_{ij}}\right|_{(\vec{x})} \right) \left/ \sum_{j'=1}^{q_i} x_{ij'} \left. \frac{\partial P}{\partial x_{ij'}}\right|_{(\vec{x})} \right.
  \end{equation*}
  Then $P(\Im(\vec{x}))>P(\vec{x})$ unless $\Im(\vec{x})=\vec{x}$.
  \label{Thm_Baum_Eagon}
\end{theorem}

The Baum-Welch  algorithm is a classic technique used to find the unknown parameters of a hidden Markov model (HMM).
A HMM describes the joint probability of a collection of ``hidden" and observed discrete random variables.  It relies on the assumption that the $i$-th hidden variable given the $(i - 1)$-th hidden variable is independent of previous hidden variables, and the current observation variables depend only on the current hidden state.
The Baum-Welch algorithm uses the well known EM algorithm to find the maximum likelihood estimate of the parameters of a hidden Markov model given a set of observed feature vectors. More detailed exposition of these ideas can be found here \cite{bilmes1998gentle}.
The probability of making a specific time series of observations of length $T$ can be shown to be a homogeneous polynomial $P$ of degree $T$ with nonnegative (integer) coefficients of the model parameters. Baum-Welch algorithm is  homologous to the iterative process derived by applying the Baum-Eagon theorem to polynomial $P$ \cite{Baum_Eagon,welch2003hidden}. 

In a nutshell, both Baum-Welch and $\text{MWU}_\ell$ in congestion games are  special cases of the Baum-Eagon iterative process (for different polynomials $P$).

\subsection{Multiplicative Weights Update}\label{ourdynamic}
\noindent
In this section, we describe the MWU dynamics (both the linear $\text{MWU}_\ell$, and the exponential $\text{MWU}_e$ variants) applied in congestion games. The update rule (function) $\xi : \Delta \to \Delta$ (where $\vec{p}(t+1) = \xi(\vec{p}(t))$) for the linear variant $\text{MWU}_\ell$ is as follows:

\begin{equation}
    p_{i\gamma}(t+1) = (\xi(\vec{p}(t)))_{i\gamma} = p_{i\gamma}(t)\frac{1-\epsilon_{i}c_{i\gamma}(t)}{1-\epsilon_{i}\hat{c}_{i}(t)}, \,\,\,{\forall i \in \mathcal{N}, \forall \gamma \in S_i},
    \label{Our_Dynamics}
\end{equation}
\noindent
where $\epsilon_{i}$ is a constant (can depend on player $i$ but not on $\vec{p}$) so that both enumerator and denominator of the fraction in (\ref{Our_Dynamics}) are positive (and thus the fraction is well defined). Under the assumption that $1/\epsilon_{i} > \frac{1}{\beta} \defeq \sup_{i,\vec{p} \in \Delta, \gamma \in S_i} \left\{c_{i \gamma} \right\}$, it follows that $1/\epsilon_{i} > c_{i\gamma}$ for all $i, \gamma$ and hence $1/\epsilon_{i} > \hat{c}_i$. 

\smallskip

The update rule (function) $\eta : \Delta \to \Delta$ (where $\vec{p}(t+1) = \eta(\vec{p}(t))$) for the exponential variant $\text{MWU}_e$ is as follows:
\begin{equation}
    p_{i\gamma}(t+1) = (\eta(\vec{p}(t)))_{i\gamma} = p_{i\gamma}(t)\frac{(1-\epsilon_{i})^{c_{i\gamma}(t)}}{\sum_{\gamma' \in S_i}p_{i\gamma'}(t)(1-\epsilon_{i})^{c_{i\gamma'}(t)}}, \,\,\,{\forall i \in \mathcal{N}, \forall \gamma \in S_i},
    \label{Our_Dynamics2}
\end{equation}
where $\epsilon_{i}<1$ is a constant (can depend on player $i$ but not on $\vec{p}$).

\noindent
\begin{remark}
Observe that $\Delta$ is invariant under the discrete dynamics \eqref{Our_Dynamics}, \eqref{Our_Dynamics2} defined above. If $p_{i\gamma}=0$ then $p_{i\gamma}$ remains zero, and if it is positive, it remains positive (both numerator and denominator are positive) and also is true that $\sum_{\gamma \in S_i}p_{i\gamma}=1$ for all agents $i$. A point $\vec{p}^{*}$  is called a fixed point if it stays invariant under the update rule of the dynamics, namely $\xi(\vec{p}^{*}) = \vec{p}^{*}$ or $\eta(\vec{p}^{*}) = \vec{p}^{*}$.
A point $\vec{p}^*$ is a fixed point of (\ref{Our_Dynamics}), (\ref{Our_Dynamics2}) if for all $i,\gamma$ with $p^*_{i\gamma}>0$ we have that $c_{i\gamma}= \hat{c}_{i}$. To see why, observe that if $p^{*}_{i\gamma}, p^{*}_{i\gamma'}>0$, then $c_{i\gamma} = c_{i\gamma'}$ and thus $c_{i\gamma} = \hat{c}_{i}$. We conclude that the set of fixed points of both dynamics (\ref{Our_Dynamics}), (\ref{Our_Dynamics2}) coincide and are supersets of the set of Nash equilibria of the corresponding congestion game.
\end{remark} 
\section{Convergence of $\textrm{MWU}_{\ell}$ to Nash Equilibria}
\label{section_Discr_Dynamics}
We first prove that  $\text{MWU}_\ell$  (\ref{Our_Dynamics}) converges to fixed points. Technically, we establish that function $\Psi \defeq \mathbb{E}_{\vec{s} \sim \vec{p}} \left[\Phi(\vec{s}) \right]$ is strictly decreasing along any nontrivial (i.e. nonequilibrium) trajectory, where $\Phi$ is the potential function of the  congestion game as defined in Section \ref{Concepts}. Formally we show the following theorem:
\begin{theorem} [$\Psi$ is decreasing]\label{mainthm}
Function $\Psi$ is decreasing w.r.t. time, \ie $\Psi(\vec{p}(t+1)) \leq \Psi(\vec{p}(t))$ where equality $\Psi(\vec{p}(t+1)) = \Psi(\vec{p}(t))$ holds \textbf{only} at fixed points.
\end{theorem}
\noindent
We define the function
\begin{equation}
    Q(\vec{p}) \defeq \underbrace{\sum_{i \in \mathcal{N}} \left ( \left(1/\epsilon_i - 1/\beta\right) \cdot \sum_{\gamma \in S_i} p_{i\gamma} \right ) + 1/\beta \cdot \prod_{i \in \mathcal{N}} \left(\sum_{\gamma \in S_i}p_{i\gamma} \right)}_{\textrm{constant term}} - \Psi(\vec{p}),
  \label{Q_definition}
\end{equation}
and show that $Q(\vec{p})$ is strictly increasing w.r.t time, unless $\vec{p}$ is a fixed point. Observe that $\sum_{\gamma \in S_i} p_{i\gamma}=1$ since $\vec{p}$ lies in $\Delta$, but we include this terms in $Q$ for technical reasons that will be made clear later in the section.
By showing that $Q$ is increasing with time, Theorem \ref{mainthm} trivially follows since $Q = const - \Psi$ where $const = \sum_{i \in \mathcal{N}}1/\epsilon_i$. To show that $Q(\vec{p})$ is strictly increasing w.r.t time, unless $\vec{p}$ is a fixed point, we use a generalization of an inequality by Baum and Eagon \cite{Baum_Eagon} on function $Q$.
\begin{corollary}[Generalization of Baum-Eagon] Theorem \ref{Thm_Baum_Eagon} holds even if $P$ is non-homogeneous. \label{thm:generalBaumEagon}
\end{corollary}
\begin{proof} We prove it by doing a reduction. Let $P(\vec{x})$ be a non-homogeneous polynomial of degree $d$ on variables $\{x_{ij}\}$ with $\vec{x} \in D$ ($D$ is a product of simplices). We introduce a dummy variable $y$ that is always set to one and $D' = \{(\vec{x},y): \vec{x} \in D, y=1\}$. We define the polynomial $P'(\vec{x},y)$ where for each monomial of $P$ with total degree $d'$ so that $d' \leq d$, we have the same monomial in $P'$ multiplied by $y^{d-d'}$. It is obvious to see that $P'$ is homogeneous of degree $d$. It is also obvious to check that the dynamics as defined in Theorem \ref{Thm_Baum_Eagon} for polynomial $P'$ remains the same as for polynomial $P$ (apart from the extra(dummy) variable $y$ which is always one) since if $y=1$ at time $t$ then at time $t+1$, $y$ is equal to $\frac{y\frac{\partial P'(\vec{x},y)}{\partial y}}{y\frac{\partial P'(\vec{x},y)}{\partial y}} = 1$, \ie $y$ indeed is always equal to one and $ \left.  \frac{\partial P'(\vec{x},y)}{\partial x_{ij}}\right|_{(\vec{x,1})}  =  \left. \frac{\partial P(\vec{x})}{\partial x_{ij}}\right|_{(\vec{x})}$.

\noindent
We conclude that Theorem \ref{Thm_Baum_Eagon} holds for non-homogeneous polynomials.
\end{proof}
We want to apply Corollary \ref{thm:generalBaumEagon} on $Q$. To do so, it suffices to show that $Q(\vec{p})$ is a polynomial with nonnegative coefficients.  
\begin{lemma}
\label{lem:polynomial}
$Q(\vec{p})$ is a polynomial with respect to $p_{i \gamma}$ and has nonnegative coefficients.
\end{lemma}
\begin{proof}
In a congestion game, the cost of the function of any player $i$ can be written as the sum of the potential function $\Phi(\textbf{s})$ and a dummy term which depends on the actions of all the rest players (not on the actions of player $i$), \ie
    \begin{equation}
        c_i(\vec{s}) = \Phi(\vec{s}) + D_i(\vec{s}_{-i}).
    \label{Potential_Game}
    \end{equation}
    By taking expectations in Equation (\ref{Potential_Game}) we get that $\hat{c}_i = \Psi + \mathbb{E}_{\vec{s}_{-i} \sim \vec{p}_{-i}}[D_i(\vec{s}_{-i})]$. Using the law of total expectation it also follows that the expected cost of player $i$ satisfies $\hat{c}_{i} = \sum_{\gamma \in S_i} p_{i\gamma} c_{i\gamma}$. Therefore $\sum_{\gamma \in S_i} p_{i\gamma} c_{i\gamma} = \Psi(\vec{p}) + \mathbb{E}_{\vec{s}_{-i} \sim \vec{p}_{-i}}[D_i(\vec{s}_{-i})]$.

We take the partial derivative of both L.H.S and R.H.S for variable $p_{i\gamma}$ and we conclude that the  following holds:
\begin{equation}
      c_{i\gamma} = \frac{\partial \Psi(\vec{p})}{\partial p_{i \gamma}} +\underbrace{\frac{\partial \mathbb{E}_{\vec{s}_{-i} \sim \vec{p}_{-i}}[D_i(\vec{s}_{-i})]}{\partial p_{i\gamma}}}_{=0},
      \textrm{ therefore } \frac{\partial Q(\vec{p})}{\partial p_{i \gamma}} = \underbrace{1/\epsilon_i - 1/\beta + 1/\beta \cdot \prod_{j\neq i} \left ( \sum_{\gamma \in S_j}p_{j\gamma} \right) - c_{i \gamma}}_{1/\epsilon_i - c_{i\gamma} \textrm{ since } \vec{p} \in \Delta}\label{eq:derivative}
\end{equation}
\noindent
Since the R.H.S of (\ref{eq:derivative}) does not depend on $p_{i\gamma}$, $Q$ is a linear function w.r.t $p_{i \gamma}$ for all $i \in \mathcal{N}, \gamma \in S_i$. Therefore, it is a polynomial of degree $N$ with respect to $\vec{p}$.

Finally, we will show that all the coefficients of  the  polynomial $Q$  are non-negative.
Let's focus on the monomials containing the term $p_{i \gamma}$ (for some $i, \gamma$).
By (\ref{eq:derivative}) we have that the summation of those monomials is equal to 
$(1/\epsilon_i - 1/\beta)p_{i \gamma} + \left ( 1/\beta \cdot  \prod_{j\neq i} \left ( \sum_{\gamma \in S_j}p_{j\gamma} \right) - c_{i \gamma} \right )  p_{i \gamma}$ which expands to $(1/\epsilon_i - 1/\beta)p_{i \gamma} + \left ( 1/\beta \cdot  \sum_{\mathbf{s}_{-i}\in\mathbf{S}_{-i} } \prod_{j \neq i} p_{j \mathbf{s}_{j}} - c_{i \gamma} \right )  p_{i \gamma}$, where  $\mathbf{S}_{-i} \defeq \times_{j\neq i}S_j$. However, we have

         \begin{equation*}
          c_{i \gamma} = \sum_{\mathbf{s}_{-i}\in \mathbf{S}_{-i} } \prod_{j \neq i} p_{j \mathbf{s}_{j}} \cdot \underbrace{\left( \sum_{e \in \gamma} c_e\left(1+k_e(\mathbf{s}_{-i})\right) \right)}_{\leq \frac{1}{\beta} \textrm{ by definition of } \beta},
         \end{equation*} where $k_e (\mathbf{s}_{-i})$ denotes the number of players apart from $i$ that choose edge $e$ in the strategy profile $\vec{s}_{-i}$.  Combining everything together we have that summation of all monomials including $p_{i \gamma}$ is equal to:
         
    $$(1/\epsilon_i - 1/\beta)p_{i \gamma} +  \Big( 1/\beta - \underbrace{\left( \sum_{e \in \gamma} c_e(1+k_e(\mathbf{s}_{-i})) \right)}_{\leq \frac{1}{\beta}} \Big) \cdot \sum_{\mathbf{s}_{-i}\in \mathbf{S}_{-i} } \prod_{j \neq i} p_{j \mathbf{s}_{j}}\cdot  p_{i \gamma}$$    
         
  Clearly, each summand has a nonnegative coefficient. Hence, each monomial containing $p_{i \gamma}$ has a nonnegative coefficient.  The above is true for all $i,\gamma$ and the claim follows.
\end{proof}

\noindent
Using Lemma \ref{lem:polynomial} and Corollary \ref{thm:generalBaumEagon} we show the following:
\begin{theorem}
Let $Q$ be the function defined in \eqref{Q_definition}. Let also $\vec{p}(t) \in \Delta$ be the point $\text{MWU}_\ell$  (\ref{Our_Dynamics}) outputs at time $t$ with update rule $\xi$. It holds that $Q(\vec{p}(t+1)) \defeq Q(\xi(\vec{p}(t)))>Q(\vec{p}(t))$ unless $\xi(\vec{p}(t))=\vec{p}(t)$ (fixed point). Namely $Q$ is strictly increasing with respect to the number of iterations $t$ unless $\text{MWU}_\ell$  is at a fixed point.
\label{Qincreasing}
\end{theorem}
\begin{proof}
  By Lemma \ref{lem:polynomial}, $Q(\vec{p})$ is a polynomial with has nonnegative coefficients. Therefore, we can apply Corollary \ref{thm:generalBaumEagon} for polynomial $Q$.  In this case, the Baum-Eagon theorem defines the  map:
  \begin{align*}
  p_{i\gamma}(t+1) &= \left( \left. p_{i\gamma}(t) \frac{\partial Q}{\partial p_{i\gamma}}\right|_{(\vec{p}(t))} \right) \left/ \sum_{\delta \in S_i} p_{i\delta} \left. \frac{\partial Q}{\partial p_{i\delta}}\right|_{(\vec{p}(t))} \right.
  \\& \overset{(\ref{eq:derivative})}{=} \frac{p_{i\gamma}(t) (1/\epsilon_i - c_{i\gamma})}{\sum_{\delta \in S_i}p_{i\delta}(t)(1/\epsilon_i - c_{i\delta})} = p_{i\gamma}(t) \frac{1/\epsilon_i - c_{i\gamma}}{1/\epsilon_i - \hat{c}_i},
  \end{align*}
  which coincides with  $\text{MWU}_\ell$ (\ref{Our_Dynamics}). Thus, it is true that $Q(\vec{p}(t+1))>Q(\vec{p}(t))$ unless $\vec{p}(t+1)=\vec{p}(t)$. This proof justifies the reason we added the term $\sum_{i \in \mathcal{N}} \left ( \left(1/\epsilon_i - 1/\beta\right) \cdot \sum_{\gamma \in S_i} p_{i\gamma} \right ) + 1/\beta \cdot \prod_{i \in \mathcal{N}} \left(\sum_{\gamma \in S_i}p_{i\gamma} \right)$ in $Q$, namely so that the partial derivatives give us $\text{MWU}_\ell$ dynamics.
\end{proof}

As stated earlier in the section, if $Q(\vec{p}(t))$ is strictly increasing with respect to time $t$ unless $\vec{p}(t)$ is a fixed point, it follows that the expected potential function $\Psi(\vec{p}(t)) = const - Q(\vec{p}(t))$ is strictly decreasing unless $\vec{p}(t)$ is a fixed point and Theorem \ref{mainthm} is proved. Moreover, we can derive the fact that our dynamics converges to fixed points as a corollary of Theorem \ref{mainthm}.
\begin{theorem}[Convergence to fixed points]\label{cormain} $\text{MWU}_\ell$ dynamics (\ref{Our_Dynamics}) converges to fixed points.
\end{theorem}
\begin{proof}
Let $\Omega \subset \Delta$ be the set of limit points of an orbit $\vec{p}(t)$. 
 $\Psi(\vec{p}(t))$ is decreasing with respect to time $t$ by Theorem \ref{mainthm} and so, because $\Psi$ is bounded on $\Delta$, $\Psi(\vec{p}(t))$ converges as $t\to \infty$  to $\Psi^{*}= \inf_t\{\Psi(\vec{p}(t))\}$. By continuity of $\Psi$ we get that $\Psi(\vec{y})= \lim_{t\to\infty} \Psi(\vec{p}(t)) = \Psi^*$ for all $\vec{y} \in \Omega$. So $\Psi$ is constant on $\Omega$. Also $\vec{y}(t) = \lim_{n \to \infty} \vec{p}(t_n + t)$ as $n \to \infty $ for some sequence of times $\{t_i\}$ and so $\vec{y}(t)$ lies in $\Omega$, i.e. $\Omega$ is invariant. Thus, if $\vec{y} \equiv \vec{y}(0) \in \Omega$ the orbit $\vec{y}(t)$ lies in $\Omega$ and so $\Psi(\vec{y}(t)) = \Psi^*$ on the orbit. But $\Psi$ is strictly decreasing except on equilibrium orbits and so $\Omega$ consists entirely of fixed points.
\end{proof}

We conclude the section by strengthening the convergence result (\ie Theorem \ref{cormain}). We show that if the initial distribution $\vec{p}$ is in the interior of $\Delta$ then we have convergence to Nash equilibria.
\begin{theorem}[Convergence to Nash equilibria] \label{thm:convergencenash} Assume that the fixed points of (\ref{Our_Dynamics}) are isolated. Let $\vec{p}(0)$ be a point in the interior of $\Delta$. It follows that $\lim_{t \to \infty} \vec{p}(t) = \vec{p}^{*}$ is a Nash equilibrium.
\end{theorem}
\begin{proof} We showed in Theorem \ref{cormain} that $\text{MWU}_\ell$ dynamics (\ref{Our_Dynamics}) converges, hence $\lim_{t \to \infty} \vec{p}(t)$ exists (under the assumption that the fixed points are isolated) and is equal to a fixed point of the dynamics $\vec{p}^*$. Also it is clear from the dynamics that $\Delta$ is invariant, \ie $\sum_{\delta \in S_j}p_{j\delta}(t)=1$, $p_{j\delta}(t)>0$ for all $j$ and $t\geq 0$ since $\vec{p}(0)$ is in the interior of $\Delta$.

Assume that $\vec{p}^*$ is not a Nash equilibrium, then there exists a player $i$ and a strategy $\gamma \in S_i$ so that $c_{i\gamma}(\vec{p}^*)<\hat{c}_i(\vec{p}^*)$ (on mixed strategies $\vec{p}^*$) and $p^* _{i\gamma} =0$. Fix a $\zeta>0$ and let $U_{\zeta}=\{\vec{p}: c_{i\gamma}(\vec{p})<\hat{c}_i(\vec{p})-\zeta\}$. By continuity we have that $U_{\zeta}$ is open. It is also true that $\vec{p}^* \in U_{\zeta}$ for $\zeta$ small enough.

Since $\vec{p}(t)$ converges to $\vec{p}^*$ as $t\to \infty$, there exists a time $t_{0}$ so that for all $t' \geq t_0$ we have that $\vec{p}(t') \in U_{\zeta}$. However, from $\text{MWU}_\ell$ dynamics (\ref{Our_Dynamics}) we get that if $\vec{p}(t')\in  U_{\zeta}$ then $1-\epsilon_i c_{i\gamma}(t')>1-\epsilon_i \hat{c}_{i}(t')$ and hence $p_{i\gamma}(t'+1) = p_{i\gamma}(t') \frac{1-\epsilon_i c_{i\gamma}(t')}{1-\epsilon_i \hat{c}_{i}(t')} \geq p_{i\gamma}(t')>0$, \ie $p_{i\gamma}(t')$ is positive and increasing with $t' \geq t_0$. We reached a contradiction since $p_{i\gamma}(t) \to p^*_{i\gamma}=0$, thus $\vec{p}^*$ is a Nash equilibrium.
\end{proof}
\section{Non-Convergence of $\textrm{MWU}_{e}$: Limit Cycle and Chaos }\label{sec:nonconvergence}
We consider a symmetric two agent congestion game with two edges $e_1,e_2$. Both agents have the same two available strategies $\gamma_1 = \{e_1\}$ and $\gamma_2 = \{e_2\}$. We denote $x,y$ the probability that the first and the second agent respectively choose strategy $\gamma_1$.

For the first example, we assume that $c_{e_1} (l) = \frac{1}{2}\cdot l$ and $c_{e_2} (l) = \frac{1}{2}\cdot l$. Computing the expected costs we get that $c_{1\gamma_1} = \frac{1+y}{2}$, $c_{1\gamma_2} = \frac{2-y}{2} $, $c_{2\gamma_1} = \frac{1+x}{2}$, $c_{2\gamma_2} = \frac{2-x}{2}$. $\text{MWU}_e$ then becomes $x_{t+1} = x_{t} \frac{(1 - \epsilon_1)^{\frac{(y_{t}+1)}{2}}}{x_t(1 - \epsilon_1)^{\frac{y_{t}+1}{2}}+(1-x_t)(1 - \epsilon_1)^{\frac{2-y_{t}}{2}}}$ (first player) and $y_{t+1} = y_{t} \frac{(1 - \epsilon_2)^{\frac{x_{t}+1}{2}}}{y_t(1 - \epsilon_2)^{\frac{x_{t}+1}{2}}+(1-y_t)(1 - \epsilon_2)^{\frac{2-x_{t}}{2}}}$ (second player). We assume that $\epsilon_1 = \epsilon_2$ and also that $x_0 = y_0$ (players start with the same mixed strategy. Due to symmetry, it follows that $x_t = y_t$ for all $t \in\mathbb{N}$, thus it suffices to keep track only of one variable (we have reduced the number of variables of the update rule of the dynamics to one) and the dynamics becomes $x_{t+1} = x_{t} \frac{(1 - \epsilon)^{\frac{x_{t}+1}{2}}}{x_t(1 - \epsilon)^{\frac{x_{t}+1}{2}}+(1-x_t)(1 - \epsilon)^{\frac{2-x_{t}}{2}}}$. Finally, we choose $\epsilon = 1 - e^{-10}$ and we get $$x_{t+1} = H(x_t) = x_{t} \frac{e^{-5(x_{t}+1)}}{x_te^{-5(x_{t}+1)}+(1-x_t)e^{-5(2-x_{t})}},$$ i.e., we denote $H(x) =  \frac{xe^{-5(x+1)}}{xe^{-5(x+1)}+(1-x)e^{-5(2-x)}}$.

For the second example, we assume that $c_{e_1} (l) = \frac{1}{4}\cdot l$ and $c_{e_2} (l) = \frac{1.4}{4}\cdot l$. Computing the expected costs we get that $c_{1\gamma_1} = \frac{1+y}{4}$, $c_{1\gamma_2} = \frac{1.4(2-y)}{4} $, $c_{2\gamma_1} = \frac{1+x}{4}$, $c_{2\gamma_2} = \frac{1.4(2-x)}{4}$. $\text{MWU}_e$ then becomes $x_{t+1} = x_{t} \frac{(1 - \epsilon_1)^{\frac{(y_{t}+1)}{4}}}{x_t(1 - \epsilon_1)^{\frac{y_{t}+1}{4}}+(1-x_t)(1 - \epsilon_1)^{\frac{1.4(2-y_{t})}{4}}}$ (first player) and $y_{t+1} = y_{t} \frac{(1 - \epsilon_2)^{\frac{x_{t}+1}{4}}}{y_t(1 - \epsilon_2)^{\frac{x_{t}+1}{4}}+(1-y_t)(1 - \epsilon_2)^{\frac{1.4(2-x_{t})}{4}}}$ (second player). We assume that $\epsilon_1 = \epsilon_2$ and also that $x_0 = y_0$ (players start with the same mixed strategy. Similarly, due to symmetry, it follows that $x_t = y_t$ for all $t \in\mathbb{N}$, thus it suffices to keep track only of one variable  and the dynamics becomes $x_{t+1} = x_{t} \frac{(1 - \epsilon)^{\frac{x_{t}+1}{4}}}{x_t(1 - \epsilon)^{\frac{x_{t}+1}{4}}+(1-x_t)(1 - \epsilon)^{\frac{1.4(2-x_{t})}{4}}}$. Finally, we choose $\epsilon = 1 - e^{-40}$ and we get $$x_{t+1} = G(x_t) = x_{t} \frac{e^{-10(x_{t}+1)}}{x_te^{-10(x_{t}+1)}+(1-x_t)e^{-14(2-x_{t})}},$$ i.e., we denote $G(x) =  \frac{xe^{-10(x+1)}}{xe^{-10(x+1)}+(1-x)e^{-14(2-x)}}$.
\subsection{Analyzing $x_{t+1} = H(x_t)$}
\label{sec:limitcycle}

\input{H_expon_bubbles}

\subsubsection{The signs of the derivative of \;$H(H(x))$}
\smallskip
In this subsection we analyze the monotonicity of $H(H(x))$.
\begin{lemma}\label{lem:monotone}
There exist numbers $0<y_0<x_0<1/2<x_1<y_1<1$ so that:
\begin{itemize}
\item For $x \in [0,y_0], [x_0,x_1] $ and $[y_1,1]$ $H(H(x))$ is strictly increasing,
\item for $x \in [y_0,x_0]$ and $x \in [x_1,y_1]$ $H(H(x))$ is strictly decreasing,
\end{itemize}
where $x_0 = \frac{1}{10}(5-\sqrt{15}) \approx 0.1127$, $x_1 = \frac{1}{10}(5+\sqrt{15}) \approx 0.8873$, $y_0 \in (0,x_0)$ so that $H(y_0) = x_0$ and $y_1 \in (x_1,1)$ so that $H(y_1) = x_1$.
\end{lemma}
\begin{proof}First of all it holds that $\frac{d H(H(x))}{dx} = H'(H(x))\cdot H'(x)$, therefore we will analyze the signs of $H'(H(x))$ and $H'(x)$ separately. Direct calculations give $H'(x) = e^{5 + 10 x} \frac{1 - 10 x + 10 x^2}{(e^{10 x} (-1 + x) - e^5 x)^2}$. The roots of $1 - 10 x + 10 x^2$ are
$x_0$ and $x_1$ (defined in the statement). We conclude that $H$ is strictly increasing in $[0,x_0]$ and $[x_1,1]$ and strictly decreasing in $[x_0,x_1]$.

Moreover $H(x_0)\approx 0.8593>x_0$ thus lies in $(1/2,x_1)$ and $H(x_1) \approx 0.1406<x_1$ and hence lies in $(x_0,1/2)$. Let $y_0 \in (0,x_0)$ so that $H(y_0) = x_0$ (since $H$ is strictly increasing in $[0,x_0]$, $H(0) =0$ and $H(x_0)>x_0$, there exists a unique $y_0$) and by similar argument let $y_1$ the unique real in $[x_1,1]$ so that $H(y_1) = x_1$.

We have the following cases:
\begin{itemize}
\item For $x \in (0,y_0)$ we get that both $H'(x)$ and $H'(H(x))$ are positive and hence $H(H(x))$ is strictly increasing in $[0,y_0]$ (area 1 of the figure \ref{fig:M1}).
\item For $x \in (y_0,x_0)$ we get that $H'(x)$ is positive and $H'(H(x))$ is negative, thus $H(H(x))$ strictly decreasing in $[y_0,x_0]$ (area 2 of the figure \ref{fig:M1}).
\item For $x\in (x_0,x_1)$ we get that $H'$ is negative and since $(H(x_1),H(x_0))\subset (x_0,x_1)$, H is monotone we have that $H'(H(x))$ is also negative, namely $H(H(x))$ is strictly increasing in $[x_0,x_1]$ (areas 3,4,5 and 6 of the figure \ref{fig:M1}).
\item For $x \in (x_1,y_1)$ we get that $H'(x)$ is positive and $H'(H(x))$ is negative and hence $H(H(x))$ is strictly decreasing in $[x_1,y_1]$ (area 7 of the figure \ref{fig:M1}).
\item For $x \in (y_1,1)$ we get that $H'(x)$ is positive and $H'(H(x))$ is positive, thus $H(H(x))$ strictly increasing in $[y_1,1]$ (area 8 of the figure \ref{fig:M1}).
\end{itemize}
\end{proof}
\subsubsection{The fixed points of $H(H(x))$}
\begin{lemma}\label{lem:fixedpoints}
$H(H(x))$ has 5 fixed points, $0<\rho_1<1/2<\rho_2 = 1-\rho_1<1$. Moreover $H(H(x)) - x$ is positive in $(0,\rho_1)$, $(1/2,\rho_2)$ and negative in $(\rho_1,1/2)$, $(\rho_2,1)$.
\end{lemma}
\begin{proof}
By direct calculations we get that
\begin{align*}
H(H(x)) &= \frac{x}{\left(e^{-5+10x}(1-x)+x\right)\left(\frac{x}{e^{-5+10x}(1-x)+x}+e^{-5+\frac{10x}{e^{-5+10x}(1-x)+x}}\left(1-\frac{x}{e^{-5+10x}(1-x)+x}\right)\right)}
\\& = \frac{x}{x+e^{10x\left(1+\frac{1}{e^{-5+10x}(1-x)+x}\right)-10}(1-x)}
\end{align*}
It is clear that $H(H(0))=0, H(H(1))=1$ and $H(H(1/2)) = 1/2$. In order to find the other fixed points, it suffices to analyze the roots of the function $1-x-e^{10x\left(1+\frac{1}{e^{-5+10x}(1-x)+x}\right)-10}(1-x)$. By cancelling the common factor $(1-x)$ (we have already take into account $x=1$), we have to analyze
$g(x) \defeq 1-e^{10x\left(1+\frac{1}{e^{-5+10x}(1-x)+x}\right)-10}$. It follows by the monotonicity of $e^x$ that $g(x)=0$ iff $10x\left(1+\frac{1}{e^{-5+10x}(1-x)+x}\right)-10=0$, i.e., $\frac{x}{e^{-5+10x}(1-x)+x}=1-x$.

To solve the equation above, it suffices to analyze the roots of the function $$g_1(x) \defeq x-(1-x)\left(e^{-5+10x}(1-x)+x \right) = x^2-e^{-5+10x}(1-x)^2.$$
By direct calculation we have to find the roots of $g_2(x) \defeq x - e^{-2.5+5x}(1-x)$ (since $0 \leq x \leq 1$). Finally, we take the derivative of $g_2$ which is $g'_2(x) = 1 + e^{-2.5+5x} - 5e^{-2.5+5x}(1-x) = 1+ e^{-2.5+5x}(5x-4)$. Clearly $g''_2(x)$ is negative in $[0,3/5)$, positive in $(3/5,1]$ and zero at $3/5$.
Also $g'_2(0)\approx 0.67>0, g'_2(3/5) \approx -0.648<0$ and $g'_2(1) >0$, i.e., by Bolzano's theorem $g'_2(x)$ has a unique root in $(0,3/5)$ (say $\alpha_1$) and a unique root in $(3/5,1)$ (say $\alpha_2$). Finally, since $g'_2(1/2) = -0.5 <0$ and $g'_2(x_0) \approx 0.504>0$, it follows that $x_0<\alpha_1 <1/2$ and since $g'_2(x_1) \approx 4.026$ we get that $1/2<\alpha_2 < x_1$. By the above and Rolle's theorem we conclude that $H(H(x))$ has at most 3 distinct fixed points apart from $0,1$. Since $g_2$ is increasing in $(0,x_0)$ and $g_2(x_0) \approx -0.015<0$, $g_2$ has no root in $(0,x_0]$.
Moreover, since $g_2(1/4) \approx 0.035 >0$, it follows that $g_2$ has a root in $(x_0,1/4)$ (say $\rho_1$). Hence $H(H(\rho_1)) = \rho_1$ and $1/2>1/4>\rho_1>x_0$. By observing that $H(1-x) = 1 - H(x)$, we get that $H(1- H(x)) = 1 - H(H(x))$ and also $H(H(1-x)) = H(1 - H(x))$, i.e.,
$$H(H(1-x)) = 1 - H(H(x)). $$ We substitute $x$ with $\rho_1$ and we get $H(H(1-\rho_1)) = 1 - H(H(\rho_1)) = 1 - \rho_1$, namely $\rho_2 \defeq 1 - \rho_1 >3/4$ is the remaining fixed point of $H(H(x))$. Whether $H(H(x))-x$ is positive or negative follows by same arguments.
See also the figure \ref{fig:M1} for a visualization of this theorem.
\end{proof}

\begin{figure}[t]
\begin{minipage}{1.0\textwidth}
\centering     
\subfigure[Exponential $\textrm{MWU}_{e}$:
Plot of function $H$ (blue) and its iterated versions $H^{2}$ (red), $H^{3}$ (yellow). Function $y(x)=x$ is also included.]{\label{fig:Hexp1} \includegraphics[scale=0.30]{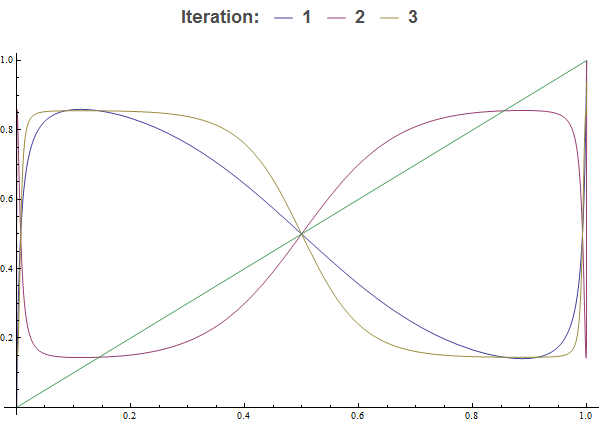}}
\;
\subfigure[Linear $\text{MWU}_\ell$:
 Plot of function $H_\ell$ (blue) and its iterated versions $H_\ell^{2}$ (red) and $H_\ell^{3}$ (yellow).
  Function $y(x)=x$ is also included.]{\label{fig:Hlin1}
  \includegraphics[scale=0.30]{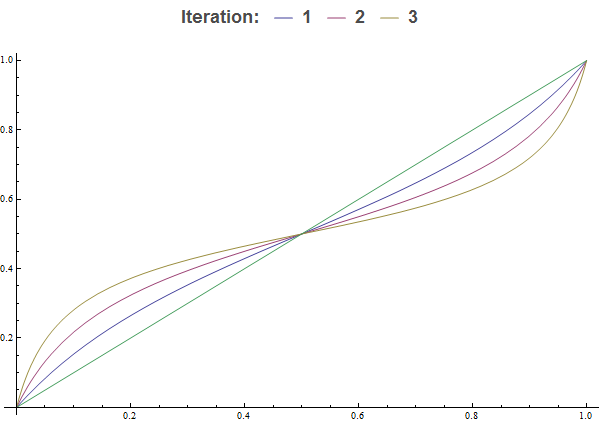}}
\end{minipage}
\end{figure}

\begin{figure}[t]
\begin{minipage}{1.0\textwidth}
\centering     
\subfigure[Exponential $\textrm{MWU}_{e}$: 
Plot of function  $H^{10}$. Function $y(x)=x$ is also included.]{\label{fig:Hexp2}\includegraphics[scale = 0.30]{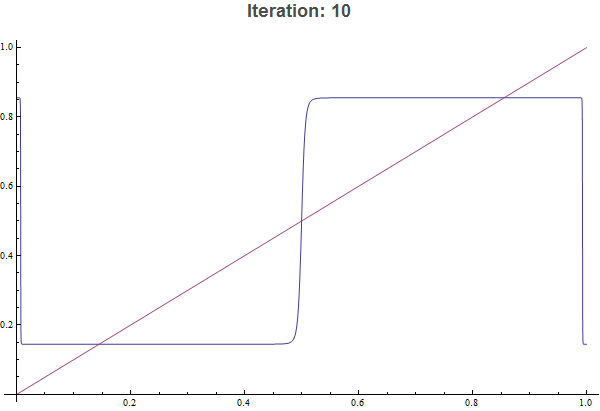}}
\;
\subfigure[Linear $\text{MWU}_\ell$:
 Plot of function  $H_\ell^{10}$. Function $y(x)=x$ is also included.]{  \label{fig:Hlin2}\includegraphics[scale = 0.30]{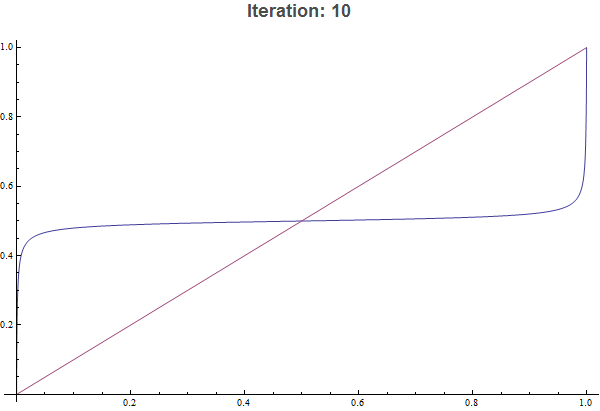}}
\end{minipage}
\caption{We compare and contrast $\textrm{MWU}_{e}$ (left) and  $\textrm{MWU}_{\ell}$ (right) in the same two agent two strategy/edges congestion game with  $c_{e_1} (l) = \frac{1}{2}\cdot l$ and $c_{e_2} (l) = \frac{1}{2}\cdot l$ and same learning rate $\epsilon=1-e^{-10}$. $\textrm{MWU}_{e}$ converges to a limit cycle whereas $\textrm{MWU}_{\ell}$ equilibrates.
Function $y(x)=x$ is also included in the graphs to help identify fixed points and periodic points.}
\end{figure}

\subsubsection{Periodic orbits}
\begin{theorem}
\label{thm:cycle}
 For all but a measure zero set $S$ of $x \in (0,1)$ we get that $\lim_{t \to \infty} H^{2t}(x) = \rho_1$ or $\rho_2$. Moreover, $H(\rho_1) = \rho_2$ and $H(\rho_2) = \rho_1$, i.e., $\{\rho_1,\rho_2\}$ is a periodic orbit. Thus, all but a measure zero set $S$ of initial conditions converge to the limit cycle $\{\rho_1,\rho_2\}$. Finally, the initial points in $S$ converge to the equilibrium $\frac{1}{2}$.
\end{theorem}
\begin{proof} Since $(\rho_1,1/2) \subset [x_0,x_1]$, from Lemma \ref{lem:monotone} it holds that $H(H(x))$ is strictly increasing in $(\rho_1,1/2)$. Thus if $\rho_1 < x< 1/2$, it follows $\rho_1 = H(H(\rho_1)) < H(H(x)) < H(H(1/2)) = 1/2$, i.e., the interval $[\rho_1,1/2]$ is invariant under $H\circ H$. Consider an initial condition $z_0 \in (\rho_1, 1/2)$ and define the sequence $z_{i+1} = H(H(z_{i})$. It is clear that $z_i \in (\rho_1,1/2)$ for all $i \in \mathbb{N}$ from previous argument. Additionally, $(z_i)_{i \in \mathbb{N}}$ is strictly decreasing because $z_{i+1} = H(H(z_{i})) < z_i$ (from Lemma \ref{lem:fixedpoints} we have $H(H(x))< x$ for all $x \in (\rho_1,1/2)$). Finally, $z_i > \rho_1$ for all $i \in \mathbb{N}$ (lower bounded), and thus the sequence converges to some limit $l$. It is easy to see that $\rho_1 \leq l < 1/2$ and also $H(H(l))=l$ by continuity of $H$, namely $l = \rho_1$ (using Lemma \ref{lem:fixedpoints}). Therefore, we showed that for any initial point $z_0 \in [\rho_1,1/2)$, we get that $\lim _{t \to \infty}H^{2t}(z_0) = \rho_1$. Analogously holds that for any initial point $z_0 \in (1/2,\rho_2]$, we get that $\lim _{t \to \infty}H^{2t}(z_0) = \rho_2$. It is clear that $\lim _{t \to \infty}H^{2t}(1/2) = 1/2$ ($1/2$ is a fixed point of $H$).

Moreover a point $z \in (x_0,\rho_1)$ we have that $z' = H(H(z)) \in (HH(x_0),HH(\rho_1))$ ($H\circ H$ is strictly increasing by Lemma \ref{lem:monotone}). Since $z< \rho_1$, we have that $z' = H(H(z)) > z$ (from Lemma \ref{lem:fixedpoints}). Therefore for any initial point $z_0 \in (x_0,\rho_1)$, the sequence $(H^{2t}(z_0))_{t \in \mathbb{N}}$ is strictly increasing and bounded by $\rho_1$, hence it converges. By similar argument as before we conclude that $\lim _{t \to \infty}H^{2t}(z_0)=\rho_1$. Analogously, it holds for any initial point $z_0 \in (\rho_2,x_1)$ that $\lim _{t \to \infty}H^{2t}(z_0)=\rho_1$.

We continue by considering the case that $z \in (y_0,x_0)$. From Lemma \ref{lem:monotone} we have that $z' = H(H(z)) \in (H(H(x_0)),H(H(y_0)))$. From Lemma \ref{lem:fixedpoints} $H(H(x_0)) >x_0 $ and $H(H(y_0)) = H(x_0) <x_1$. Therefore $z' \in (x_0,x_1)$ and from the previous cases we have that $\lim _{t \to \infty} H^{2t} (z) = \rho_1$ or $\rho_2$, unless $z' = 1/2$, i.e., unless $H(H(z))=1/2$. It is completely analogous the case $z \in (x_1,y_1)$.

To finish the proof, assume $z_0 \in (0,y_0)$. From Lemma \ref{lem:monotone} is holds that $z_1 = H(H(z_0))>z_0$. Let $n$ be the minimum index for $t$ so that $z_n = H^{2n}(z_0) > y_0$ ($n$ exists and is finite, otherwise the sequence $(H^{2t})_{t \in\mathbb{N}}$ would converge to a fixed point, which is contradiction because there is no fixed point in $(0,y_0)$). It is clear that $z_{n-1} < y_0$ and hence $$y_0< H(H(z_{n-1})) < H(H(y_0)) = H(x_0) <x_1.$$ So either $z_n = 1/2$ or $H(H(z_n))=1/2$ or else the sequence $H^{2t}$ converges to $\rho_1$ or $\rho_2$ (by reduction to the previous cases). Completely analogous is the remaining case $z_0 \in (y_1,1)$.

Therefore we showed the following: For all $z \in (0,1)$, either there exists a number $k \in \mathbb{N}$ so that $H^{2k}(z) = \frac{1}{2}$ or the limit $\lim_{t \to \infty}H^{2t}(z)$ exists and is equal to $\rho_1$ or $\rho_2$. Finally, the set $\{z \in(0,1) : \exists k \in \mathbb{N} \textrm{ s.t }H^{2k}(z) = \frac{1}{2}\}$ has measure zero (from Lemma \ref{lem:monotone}, the set $\{z:H(H(z))=1/2\}$ has cardinality at most 5). See also figure \ref{fig:Hexp2} for a visualization of the theorem. In contrast, figure \ref{fig:Hlin2} shows that the linear variant converges to the fixed point $1/2$ ($x = 1/2, y = 1/2$ is a Nash equilibrium of the corresponding game, i.e., the first example of Section \ref{sec:nonconvergence}).
\end{proof}

\subsection{Analyzing $x_{t+1} = G(x_t)$}
\begin{lemma}\label{lem:threefixedpoints} $G$ has 3 fixed points $0< 3/4 < 1$ in $[0,1]$.
\label{lem:chaos}
\end{lemma}
\begin{proof} Let $x$ be a fixed point of $G$. If $x \neq 0,1$ then $1+x = \frac{14}{10} (2-x)$, therefore $x = \frac{3}{4}$.
\end{proof}
\begin{lemma} \label{lem:threecycle} There exist a $y \in [0,1]$ so that $G(G(G(y)))=y$, $G(y) \neq y$, $G(G(y)) \neq y$ and $G(G(y)) \neq G(y)$. Hence $y,G(y),G(G(y))$ is a periodic orbit of length three.
\end{lemma}
\begin{proof} It holds that $G(G(G(0.4))) - 0.4 \approx -0.158$ and $G(G(G(0.5)))-0.5 \approx 0.496$ and hence by Bolzano's theorem there exists a $y \in (0.4,0.5)$ so that $G(G(G(y)))=y$. Observe that $y$ cannot be a fixed point of $G$ because of Lemma \ref{lem:threefixedpoints}. If $G(G(y)) = y$, then by applying $G$ we get $G(G(G(y))) = G(y)$ and hence $y = G(y)$ (contradiction since $y$ cannot be a fixed point). Finally, if $G(G(y)) = G(y)$ then by applying $G\circ G$ we get $G(G(G(G(y)))) = G(G(G(y)))$, and since $G(G(G(y)))=y$ we have that $G(y) = y$ (contradiction again). See also figure \ref{fig:Hexp3} for a visualization of the theorem.
\end{proof}
\begin{corollary}
\label{coro:chaos}
There exist two player two strategy symmetric congestion games such that $\textrm{MWU}_{e}$  has periodic orbits of length $n$  for any natural number $n>0$ and as well as an uncountably infinite set of ``scrambled" initial conditions 
(Li-Yorke chaos).
\end{corollary}
\begin{proof} It follows from Li-Yorke theorem (Theorem \ref{thm:liyorke}) and Lemma \ref{lem:threecycle}. See also figure \ref{fig:Hexp4} for a visualization of the theorem. In contrast, figure \ref{fig:Hlin4} shows that the linear variant converges to the fixed point $3/4$ ($x = 3/4, y = 3/4$ is a Nash equilibrium of the corresponding game, i.e., the second example of Section \ref{sec:nonconvergence}).
\end{proof}

\begin{figure}[t]
\begin{minipage}{1.0\textwidth}
\centering     
\subfigure[Exponential $\textrm{MWU}_{e}$: 
Plot of function $G$ (blue) and its iterated versions $G^{2}$ (red), $G^{3}$ (yellow). Function $y(x)=x$ is also included.]{\label{fig:Hexp3}\includegraphics[scale = 0.30]{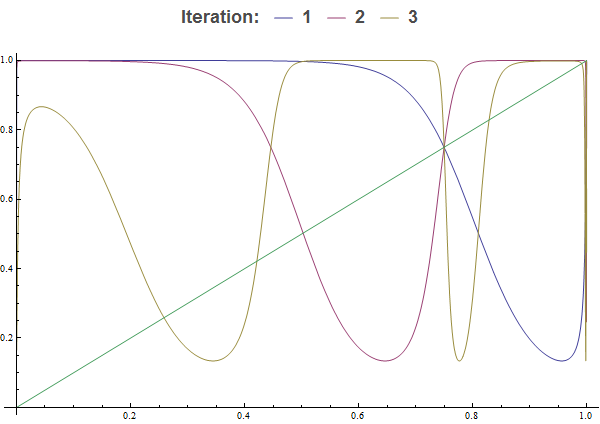}}
\;
\subfigure[Linear $\text{MWU}_\ell$: 
Plot of function $G_\ell$ (blue) and its iterated versions $G_\ell^{2}$ (red) and $G_\ell^{3}$ (yellow). Function $y(x)=x$ is also included.]{  \label{fig:Hlin3}\includegraphics[scale = 0.30]{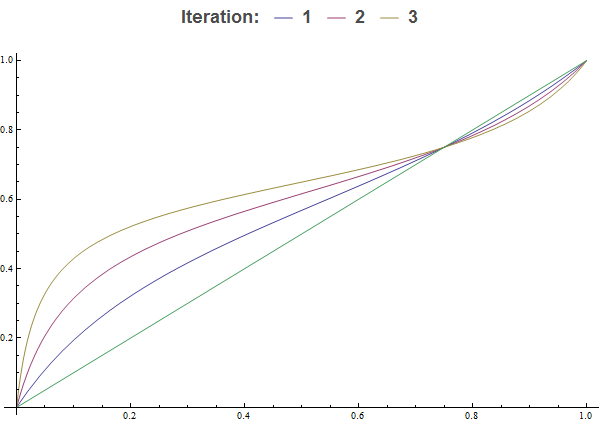}}
\end{minipage}
\end{figure}

\begin{figure}[t]
\begin{minipage}{1.0\textwidth}
\centering     
\subfigure[Exponential $\textrm{MWU}_{e}$:
Plot of function $G^{10}$. Function $y(x)=x$ is also included.]{\label{fig:Hexp4} \includegraphics[scale = 0.30]{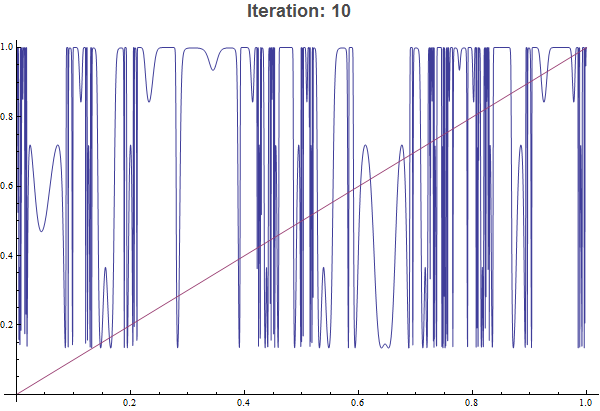}}
\;
\subfigure[Linear $\text{MWU}_\ell$:
 Plot of function $G_\ell^{10}$. Function $y(x)=x$ is also included.]{\label{fig:Hlin4}
  \includegraphics[scale=0.30]{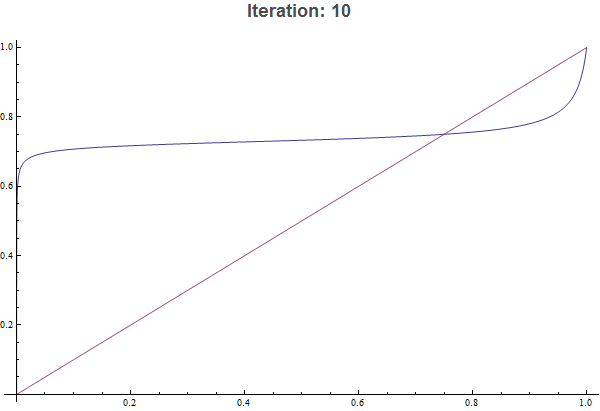}}
\end{minipage}
\caption{We compare and contrast $\textrm{MWU}_{e}$ (left) and  $\textrm{MWU}_{\ell}$ (right) in the same two agent two strategy/edges congestion game with  $c_{e_1} (l) = \frac{1}{4}\cdot l$ and $c_{e_2} (l) = \frac{1.4}{4}\cdot l$ and same learning rate $\epsilon=1-e^{-40}$. $\textrm{MWU}_{e}$ exhibits sensitivity to initial conditions whereas $\textrm{MWU}_{\ell}$ equilibrates. Function $y(x)=x$ is also included in the graphs to help identify fixed points and periodic points.}
\end{figure}

\section{Conclusion and Future Work}

We have analyzed $\text{MWU}_\ell$ in congestion games where agents use \textit{arbitrary admissible
constants} as 
 learning rates $\epsilon$
 and showed convergence to \textit{exact Nash equilibria}. We have also shown that this result is not true
for the nearly homologous exponential variant $\text{MWU}_{e}$
 even for the simplest case of two-agent, two-strategy load balancing games. There we prove that  such dynamics can provably lead to limit cycles or even chaotic behavior.

For a small enough learning rate $\epsilon$ the behavior of $\text{MWU}_{e}$ approaches that of its smooth variant, replicator dynamics, and hence convergence is once again guaranteed \cite{Kleinberg09multiplicativeupdates}. This means that as we increase the learning rate $\epsilon$ from near zero values we start off with a convergent system and we end up with a chaotic one. Numerical experiments establish that between the convergent region and the chaotic region there exists a range of values for $\epsilon$ for which the system exhibits periodic behavior. Period doubling is known as standard route for 1-dimensional chaos (e.g. logistic map) and is characterized by unexpected regularities such as the Feigenbaum constant~\cite{Strogatz2000}. Elucidating these connections is an interesting open problem. More generally, what other type regularities can be established in these non-equilibrium systems?

Another interesting question has to do with developing a better understanding of the set of conditions that result to non-converging trajectories. So far, it has been critical for our non-convergent examples that the system starts from a symmetric initial condition. Whether such irregular $\text{MWU}_{e}$ trajectories can be constructed for generic initial conditions, possibly in larger congestion games, is not known. Nevertheless, the non-convergent results, despite their non-generic nature are rather useful since they imply that we cannot hope to leverage the power of Baum-Eagon techniques for $\text{MWU}_{e}$. 
In conclusion, establishing  generic (non)convergence results (e.g. for most initial conditions, most congestion games) for $\text{MWU}_{e}$  with constant step size is an interesting future direction.


\section*{Acknowledgements}
Georgios Piliouras would like to thank Ioannis Avramopoulos for introducing him to the Li-Yorke literature.
Gerasimos Palaiopanos would like to acknowledge a SUTD Presidential fellowship.
Ioannis Panageas would like to acknowledge a MIT-SUTD postdoctoral fellowship.
Georgios Piliouras would like to acknowledge
	SUTD grant SRG ESD 2015 097 and MOE AcRF Tier 2 Grant  2016-T2-1-170. Part of this work was completed while Ioannis Panageas was a PhD student at Georgia Institute of Technology.
	Part of the work was completed  while Ioannis Panageas and Georgios Piliouras   were visiting scientists at the Simons Institute for the Theory of Computing.
	Part of the work was completed while  Georgios Piliouras was a visiting scientist at the Hausdorff Research Institute for Mathematics (HIM) during the
	Trimester Program on  Combinatorial Optimization.

\bibliographystyle{plain}
\bibliography{ms}

\end{document}